\documentclass[final,12pt]{article}
\usepackage{hyperref}
\usepackage{amssymb,amsmath,amsthm}
\usepackage{enumerate}
\usepackage[conditional,light,first,bottomafter]{draftcopy}
\draftcopyName{DRAFT\space\today}{130}
\draftcopySetScale{65}
\usepackage[letterpaper,hmargin=3.3cm,vmargin={2.6cm,3.2cm}]{geometry}
\usepackage{setspace}
\makeatletter
\renewcommand{\section}{\@startsection%
{section}%
{1}%
{0em}%
{1.7em}%
{1.2em}%
{\normalfont\large\centering\bfseries}}
\renewcommand{\@seccntformat}[1]%
{\csname the#1\endcsname.\hspace{0.5em}}
\makeatother

\numberwithin{equation}{section}

\newtheorem{theorem}{Theorem}[section]

\newtheorem{lemma}{Lemma}[section]
\newtheorem{corollary}{Corollary}[section]
\theoremstyle{definition}
\newtheorem{definition}{Definition}
\newtheorem{remark}{Remark}
\newtheorem*{notation}{Notation}

\newtheorem*{acknowledgments}{Acknowledgments}
\newcommand{\reals}{\mathbb{R}}
\newcommand{\nats}{\mathbb{N}}
\newcommand{\integers}{\mathbb{Z}}
\newcommand{\complex}{\mathbb{C}}

\newcommand{\norm}[1]{\left\|#1\right\|}

\newcommand{\inner}[2]{\left\langle#1,#2\right\rangle}

\def\ocirc#1{\ifmmode\setbox0=\hbox{$#1$}\dimen0=\ht0 \advance\dimen0
  by1pt\rlap{\hbox to\wd0{\hss\raise\dimen0
  \hbox{\hskip.2em$\scriptscriptstyle\circ$}\hss}}#1\else {\accent"17 #1}\fi}

\DeclareMathOperator{\dom}{dom}

\begin{document}
\begin{titlepage}
\title{Spectral measures of Jacobi operators with random potentials
\footnotetext{%
Mathematics Subject Classification(2000):
47B36, 
47A25, 
39A12.} 
\footnotetext{%
Keywords:
Spectral measures;
Jacobi matrices;
Random potential}
}
\author{
\textbf{Rafael del Rio and Luis O. Silva}
\\[6mm]
\small Departamento de M\'{e}todos Matem\'{a}ticos y Num\'{e}ricos\\[-1.6mm]
\small Instituto de Investigaciones en Matem\'aticas Aplicadas y en Sistemas\\[-1.6mm]
\small Universidad Nacional Aut\'onoma de M\'exico\\[-1.6mm]
\small C.P. 04510, M\'exico D.F.\\[1mm]
\small\texttt{delrio@leibniz.iimas.unam.mx}\\[-1mm]
\small\texttt{silva@leibniz.iimas.unam.mx}
}
\date{}
\maketitle
\vspace{4mm}
\begin{center}
\begin{minipage}{5in}
  \centerline{{\bf Abstract}} \bigskip Let $H_\omega$ be a
  self-adjoint Jacobi operator with a potential sequence
  $\{\omega(n)\}_n$ of independently distributed random variables with
  continuous probability distributions and let $\mu_\phi^\omega$ be
  the corresponding spectral measure generated by $H_\omega$ and the
  vector $\phi$.  We consider sets $\mathcal{A}(\omega)$ which depend
  on $\omega$ in a particular way and prove
  that $\mu_\phi^\omega(\mathcal{A}(\omega))=0$ for almost every
  $\omega$.  This is applied to show equivalence relations between
  spectral measures for random Jacobi matrices and to study the
  interplay of the eigenvalues of these matrices and their
  submatrices.

\end{minipage}
\end{center}
\thispagestyle{empty}
\end{titlepage}
\section{Introduction}
\label{sec:intro}
Let $H_0$ be a Jacobi operator with zero main diagonal in a
Hilbert space with an orthonormal basis $\{\delta_k\}_{k\in I}$, where
$I$ is a finite or
countable index set. We consider
the random self-adjoint operator given by
\begin{equation*}
 H_{\omega}=H_0 +\sum_{n\in I}\omega(n)\inner{\delta_n}{\cdot}\delta_n\,,
\end{equation*}
where $\omega(n)$ are independent random variables with continuous
(may be singular) probability distributions.

It is a well known fact regarding Schr\"odinger and Jacobi operators
with ergodic potentials, that the probability of a given
$\lambda\in\reals$ being an eigenvalue is zero
\cite{MR1102675,MR0883643,MR1223779}.  Here we present an extended
result (Theorem \ref{thm:spectral-measure-sequence}) for $H_\omega$,
which is not necessarily ergodic, when the point $\lambda$ depends on
the sequence $\omega$ except for two entries $\omega(n_0)$ and
$\omega(n_0+1)$, $n_0\in I$. This is complemented by Theorem
\ref{thm:measurable-function} when $\lambda$ is a measurable function
of $\omega$.  Since $\lambda$ is allowed to depend on $\omega$,
it is possible to apply these results to obtain information about the
spectral behavior of the above mentioned operators.

As a first application, we study equivalence relations of spectral
measures
$\mu_n^\omega(\cdot):=\inner{\delta_n}{E_{H_\omega}(\cdot)\delta_n}$,
where $E_{H_\omega}$ is the family of spectral projections for
$H_\omega$ given by the spectral theorem.  By applying
Theorems~\ref{thm:spectral-measure-sequence} and
\ref{thm:measurable-function}, we obtain
equivalence of spectral measures for one-sided infinite random Jacobi
matrices with continuous (could be singular) probability
distributions, that is, $\mu_n^\omega\sim\mu_m^\omega$ for
a.\,e. $\omega$ and any $n,m$ in $I$. When these
distributions are not only continuous but absolutely continuous, the
equivalence of spectral measures was proven in \cite{MR1779620} with
different methods.  For spectral measures of double-sided infinite
Jacobi operators, the equivalence relations
$\mu_k^\omega+\mu_l^\omega\sim\mu_m^\omega+\mu_n^\omega$ for a.\,e. $\omega$
and any $k,l,m,n\in I$ are established.

A second application concerns the interplay of the eigenvalues of
Jacobi matrices and their submatrices. This has been studied in the
context of orthogonal polynomials, in particular, there are results
describing the behavior of eigenvalues of submatrices near a
neighborhood of an eigenvalue of the whole matrix \cite{MR1971777}
\cite[Sec.\,1.2.11]{MR2105088}.  Here we show, as a consequence of
Theorems~\ref{thm:spectral-measure-sequence} and
\ref{thm:measurable-function}, that eigenvalues of a Jacobi matrix do
not coincide with eigenvalues, moments or entries of its submatrices
almost surely.  Thus, it is not only true that one point is eigenvalue
of $H_\omega$ for at most a set of zero measure as mentioned above,
but an arbitrary eigenvalue of any submatrix (which depends on
$\omega$) is not an eigenvalue of $H_\omega$ almost surely.



This work is organized as follows. In Section~\ref{sec:preliminaries}
the notation is introduced along with some preliminary
concepts. Section~\ref{sec:main} is devoted to the proof of the main
results (Theorems \ref{thm:spectral-measure-sequence} and
\ref{thm:measurable-function}), where measurability conditions
play a key role. In Section 4, we apply the results of
the previous section to study equivalence relations between
spectral measures and the possible coincidence of eigenvalues with
sets of real numbers associated with submatrices.
\section{Preliminaries}
\label{sec:preliminaries}
In this section we fix the notation and introduce the setting of the
model. Mainly we use a notation similar to that in
\cite{MR1711536}. Fix $n_1,n_2$ in
$\integers\cup\{+\infty\}\cup\{-\infty\}$  define an
interval $I$ of $\integers$ as follows
\begin{equation*}
  I:=\{n\in\integers:n_1<n<n_2\}\,.
\end{equation*}
The linear space of $M$-valued sequences $\{\xi(n)\}_{n\in
  I}$ will be denoted by $l(I,M)$, that is,
\begin{equation*}
  l(I,M):=\{\xi:I\to M\}\,.
\end{equation*}
If $M$ is itself a Hilbert space, then one has a Hilbert space
\begin{equation*}
   l^2(I,M):=\{u\in l(I,M):\sum_{n\in I}\norm{\xi(n)}_M^2<\infty\}\,,
\end{equation*}
with inner product given by
\begin{equation*}
  \inner{\xi}{\eta}:=\sum_{n\in I}\inner{\xi(n)}{\eta(n)}_M\,.
\end{equation*}

Now, let us introduce a measure in $l(I,\reals)$ as follows. Let
$\{p_n\}_{n\in I}$ be a sequence of arbitrary probability measures on
$\mathbb{R}$  and consider the product measure $\mathbb{P}
=\mathop{\times}_{n\in I}p_n$ defined on the product $\sigma$-algebra
$\mathcal{F}$ of $l(I,\reals)$ generated by the cylinder sets, i.\,e,
by sets of the form $\{\omega:\omega(i_1)\in A_1,\dots,\omega(i_n)\in
A_n\}$ for $i_1,\dots, i_n\in I$, where $A_1,\dots,A_n$ are Borel sets in
$\reals$. We have thus constructed a measure space
$\Omega=(l(I,\reals),\mathcal{F},\mathbb{P})$.

Consider $a\in l(I,\reals)$ with $a(n)>0$ for all $n\in I$, and
$\omega\in\Omega$. Define, for $\xi\in l^2(I,\mathbb{C})$,
\begin{equation}
  \label{eq:main-difference-cases}
  (H\xi)(n):=
\begin{cases}
\omega(n)\xi(n)+a(n)\xi(n+1) & n=n_1+1,\quad n_1> -\infty,\\
(\tau \xi)(n) & n_1+1<n<n_2-1,\\
a(n-1)\xi(n-1)+ \omega(n)\xi(n) & n=n_2-1,\quad n_2< +\infty,
\end{cases}
\end{equation}
where
\begin{equation}
  \label{eq:tau-difference}
  (\tau \xi)(n):=a(n-1)\xi(n-1)+\omega(n)\xi(n)+a(n)\xi(n+1)\,.
\end{equation}
In the Hilbert space $ l^2(I,\mathbb{C})$, one can uniquely associate
a closed symmetric operator with $H$ (see \cite[Sec. 47]{MR1255973})
which we shall denote by $H_{\omega}$ to emphasize the dependence on
the sequence $\omega\in\Omega$. The operator $H_{\omega}$ is a Jacobi
operator having a Jacobi matrix as its matrix representation with
respect to the canonical basis $\{\delta_k\}_{k\in I}$ in
$l^2(I,\mathbb{C})$, where
\begin{equation}
  \label{eq:canonical}
  \delta_k(n)=
  \begin{cases}
    0 & n\ne k\\
    1 & n=k\,.
  \end{cases}
\end{equation}
$H_\omega$ is defined so that $\{\delta_k\}_{k\in
  I}\subset\dom(H_\omega)$.

As in the case of differential equations, one defines the Wronskian
associated with the difference equation
(\ref{eq:main-difference-cases}) by
\begin{equation*}
  W_n(\xi,\eta):=a(n)((\xi(n)\eta(n+1)-\eta(n)\xi(n+1))\,,
\qquad n_1< n<n_2-1\,.
\end{equation*}
It turns out that, for all $n,m$ such that $n_1< m<n<n_2-1$, the Green
formula (see \cite[Eq. \,1.20]{MR1711536}) holds
\begin{equation}
  \label{eq:green}
  \sum_{k=m+1}^n(\xi(\tau \eta)-(\tau \xi)\eta)(k)=W_n(\xi,\eta)-
W_m(\xi,\eta)\,.
\end{equation}
Besides this formula, the Wronskian shares some properties with the
Wronskian of the theory of differential equations, in particular, if
$W_n(\xi,\eta)=0$ for all $n$ in a subinterval of $I$, then $\xi$ and
$\eta$ are linearly dependent in that subinterval. This is verified
directly from the definition of the Wronskian.

 Now, assume that $I=\integers$ and consider the second-order
 difference equation
\begin{equation}
  \label{eq:recurrence-spectral}
  (\tau u)(n)=zu(n)\,,\qquad n\in\integers,\,z\in\complex\,,
\end{equation}
where $\tau$ is defined in (\ref{eq:tau-difference}). Fix
$m\in\integers$ and $z\in\complex$, and take the sequences $c_m(z),s_m(z)\in
l(\integers,\complex)$ being solutions of (\ref{eq:recurrence-spectral}) and
 satisfying the following initial conditions:
\begin{align}
  c_m(z,m-1)=1\,,&\qquad c_m(z,m)=0\,,\label{eq:initial-c}\\
  s_m(z,m-1)=0\,,&\qquad s_m(z,m)=1\,.\label{eq:initial-s}
\end{align}
Because of the linear independence of $c_m(z),s_m(z)$, they constitute a
fundamental system of solutions of
(\ref{eq:recurrence-spectral}). Note that for any
$n\in\integers$, $c_m(z,n),s_m(z,n)$ are polynomials of $z$. The roots
of these polynomials are measurable functions of $\omega$.

By means of the polynomials defined above we state the
following  result \cite{MR1711536},
\cite[Prop.\,A.1]{MR1616422}.
\begin{lemma}
  \label{lem:polynomials-canonical-basis}
Consider the operator $H_\omega$
with fixed $\omega\in\Omega$. For any fixed $n\in I$, we have
\begin{equation}
\label{eq:polynimials-canonical-basis}
  \delta_n=
  \begin{cases}
    s_{n_1+1}(H_\omega,n)\delta_{n_1+1} & -\infty<n_1\\
    c_{n_2}(H_\omega,n)\delta_{n_2-1} & n_2<+\infty\\
    s_{m+1}(H_\omega,n)\delta_{m+1} + c_{m+1}(H_\omega,n)\delta_{m}  &
-\infty=n_1,n_2=+\infty\quad \forall m\in I\,.
  \end{cases}
\end{equation}
\end{lemma}
The symmetric operator $H_\omega$ is not always self-adjoint. However,
in this work, we always consider $H_\omega$ to be a self-adjoint
operator for each $\omega\in\Omega$.  If one of the numbers $n_1,n_2$
is not finite, conditions for self-adjointness should be assumed. For
instance, when both $n_1$ and $n_2$ are infinite, the so called
Carleman criterion (cf. \cite[Chap.\,7 Sec.\,3.2]{MR0222718})
\begin{equation}
  \label{eq:carleman-gen-condition}
\sum_{n\in\nats}\frac{1}{\max\{a(-n-1),a(n-1)\}}=\infty
\end{equation}
entails self-adjointness of $H_{\omega}$.

Notice that the operator $H_{\omega}$ can be written as
\begin{equation*}
 H_{\omega}=H_0 +\sum_{n\in I}\omega(n)\inner{\delta_n}{\cdot}\delta_n\,,
\end{equation*}
where $H_0$ is a self-adjoint Jacobi operator with zero main diagonal.

For the self-adjoint operator $H_\omega$, we have the following
remarks.
\begin{remark}
  \label{rem:self-adjoint-wronskian}
  For every pair $\xi,\eta$ in the domain of the self-adjoint operator
  $H_\omega$,
\begin{equation*}
  \lim_{n\to\infty}W_n(\xi,\eta)=0
\end{equation*}
(see \cite[Sec.\,2.6]{MR1711536}).
\end{remark}
\begin{remark}
  \label{rem:simplcity-multiplicity}
  From (\ref{eq:polynimials-canonical-basis}), it follows that a
  self-adjoint Jacobi operator, whose corresponding matrix is finite
  or one-sided infinite, has simple spectrum (see
  \cite[Sec.\,69]{MR1255973}). Moreover, the last equation in
  (\ref{eq:polynimials-canonical-basis}) shows that, when both $n_1,n_2$ are
  infinite, two consecutive elements of the canonical
  basis constitute a generating basis for $H_\omega$ (see
  \cite[Sec. 72]{MR1255973}).
\end{remark}

Let $\mu_{\phi}^\omega$ be the spectral measure for $H_{\omega}$ and
the vector $\phi$, viz., the unique Borel measure on $\reals$ such that
\begin{equation*}
  \inner{\phi}{f(H_{\omega})\phi}=\int_\reals f(\lambda)d\mu_{\phi}^\omega(\lambda)
\end{equation*}
for any bounded function $f$. Equivalently,
\begin{equation}
  \label{eq:mu-through-spectral-projectors}
  \mu_{\phi}^\omega(\cdot)=\inner{\phi}{E_{H_{\omega}}(\cdot)\phi}\,,
\end{equation}
where $E_{H_{\omega}}$ is the family of spectral projections for
$H_{\omega}$ given by the spectral theorem.
\begin{notation}
Below, we shall repeatedly deal
with $\mu_{\delta_n}^\omega$ (see (\ref{eq:canonical})) and we denote
it by $\mu_{n}^\omega$ for short.
\end{notation}
\begin{definition}
  Given two measures $\nu$ and $\mu$ with the same collection of
  measurable sets, we say that $\mu$ is absolutely continuous with
  respect to $\nu$, denoted $\mu\prec\nu$, if for every measurable
  $\Delta$ such that $\nu(\Delta)=0$, it follows that
  $\mu(\Delta)=0$. Also, $\nu$ and $\mu$ are said to be equivalent,
  denoted $\nu\sim\mu$, if they are mutually absolutely continuous,
  that is, if they have the same zero sets.
\end{definition}

Suppose that at least one of the numbers $n_1,n_2$ is finite. By
inserting (\ref{eq:polynimials-canonical-basis}) into
(\ref{eq:mu-through-spectral-projectors}), one obtains, for an
arbitrary Borel set $\Delta\subset\reals$ \cite[Cor.\,A.2]{MR1616422},
\begin{equation}
  \label{eq:mu-n-mu-semi}
  \mu_{n}^{\omega}(\Delta)=
  \begin{cases}
    \int_\Delta s_{n_1+1}^2(\lambda,n)d\mu_{n_1+1}^{\omega}(\lambda) &
    n_1>-\infty\\
    \int_\Delta c_{n_2}^2(\lambda,n)d\mu_{n_2-1}^{\omega}(\lambda) &
    n_2<+\infty\,.
  \end{cases}
\end{equation}

When both numbers $n_1,n_2$ are infinite, let us define, for any Borel
$\Delta\subset\reals$ and $n\in\integers$, the matrix
 \begin{equation*}
   \boldsymbol{\mu}_n(\Delta):=
  \begin{pmatrix}
  \mu_{n}^{\omega}(\Delta) &
\inner{E_{H_\omega}(\Delta)\delta_n}{\delta_{n+1}} \\[1mm]
  \inner{E_{H_\omega}(\Delta)\delta_{n+1}}{\delta_n} &
  \mu_{n+1}^{\omega}(\Delta)
  \end{pmatrix}\,.
 \end{equation*}
The third equation in (\ref{eq:polynimials-canonical-basis}) implies
\begin{equation}
  \label{eq:mu-n-mu-matrix}
  \mu_{n}^{\omega}(\Delta)=
\int_\Delta\inner{d\boldsymbol{\mu}_m(\lambda)
\begin{pmatrix}
c_{m+1}(\lambda,n)\\
s_{m+1}(\lambda,n)
\end{pmatrix}
}{
\begin{pmatrix}
c_{m+1}(\lambda,n)\\
s_{m+1}(\lambda,n)
\end{pmatrix}
}_{\complex^2}\,.
\end{equation}
There exists a matrix (see comment after \cite[Lem.\,B.13]{MR1711536})
\begin{equation*}
  \boldsymbol{R}_m(\lambda)=
  \begin{pmatrix}
    a_m(\lambda) & b_m(\lambda)\\[1mm]
    b_m(\lambda) & 1-a_m(\lambda)
  \end{pmatrix}
\end{equation*}
such that
\begin{equation}
  \label{eq:absolute-continuity}
  \boldsymbol{\mu_m}(\Delta)=
\int_\Delta\boldsymbol{R}_m(\lambda)
d(\mu_{m}^{\omega}+\mu_{m+1}^{\omega})(\lambda)\,.
\end{equation}
\begin{remark}
\label{rem:equivalence}
Notice that from Remark~\ref{rem:simplcity-multiplicity} and
\cite[Sec.\,72]{MR1255973} and (\ref{eq:absolute-continuity}) it
follows that
$\mu_{k}^{\omega}+\mu_{k+1}^{\omega}\sim\mu_{l}^{\omega}+\mu_{l+1}^{\omega}$
for any $k,l\in\integers$.
\end{remark}
\section{Main results}
\label{sec:main}
Under the assumption that $H_\omega$ is ergodic, it is well known that
a fixed $r\in\reals$ is an eigenvalue of $H_\omega$ with probability
zero \cite[Thm.2.12]{MR1223779}, \cite[Prop.V.2.8]{MR1102675}
\cite[Thm. 9.5]{MR0883643}.  In the case of $H_\omega$ considered
here, the following result holds.
\begin{theorem}
  \label{thm:spectral-measure-sequence}
  Assume that $I$ contains at least three integers and suppose $n_0,n_0+1$
  are in $I$. Let the measures $p_{n_0},p_{n_0+1}$ be continuous (a
  continuous measure evaluated at a single point of $\reals$ equals
  zero).  Consider a finite or infinite sequence of real functions
  $\{r\}_k$ ($r_k:\Omega\to\reals$), not necessarily measurable, such
  that, for $\omega,\widetilde{\omega}\in\Omega$,
\begin{equation}
  \label{eq:property-of-r}
  r_k(\omega)=r_k(\widetilde{\omega})
\end{equation}
whenever $\omega(n)=\tilde{\omega}(n)$ for all $n\in
I\setminus\{n_0,n_0+1\}$. For any  non-zero
  element $\phi$ in the Hilbert space $l^2(I,\complex)$, either
\begin{equation}
  \label{eq:mu-is-zero}
  \mu_{\phi}^\omega(\cup_kr_k(\omega))=0
\end{equation}
for $\mathbb{P}$ a. e. $\omega$, or the set of $\omega$ where
(\ref{eq:mu-is-zero}) holds is not measurable.
\end{theorem}
\begin{proof}
We consider two cases:

A) One of the numbers $n_1,n_2$ is
finite.

Without loss of generality let us assume that $n_1$ is finite.
By Remark~\ref{rem:simplcity-multiplicity}, $\delta_{n_1+1}$ is a cyclic
vector of $H_\omega$ for any $\omega\in\Omega$.

  Fix an element $r_{k_0}$ of the sequence $\{r_k\}_k$.
 Define the set
  \begin{equation*}
    \mathcal{Q}^{r_{k_0}}:=\{\omega\in\Omega:
 \mu_{n_1+1}^{\omega}(\{r_{k_0}(\omega)\})>0\}\,.
  \end{equation*}
  Let us construct a partition of $\mathcal{Q}^{r_{k_0}}$. If
  $\omega_0\in\mathcal{Q}^{r_{k_0}}$, then $r_{k_0}(\omega_0)$ is an
  eigenvalue of $H_{\omega_0}$ with corresponding eigenvector
  $\psi=E_{H_{\omega_0}}(\{r_{k_0}(\omega_0)\})\delta_{n_1+1}$. Due to
  the cyclicity of $\delta_{n_1+1}$, the converse is true, that is, if
  we have an eigenvalue $r$ of $H_{\omega_0}$, then
  $\mu_{n_1+1}^{\omega_0}(\{r\})>0$.

Analogously, if $\omega_0+t\delta_{n_0}\in\mathcal{Q}^{r_{k_0}}$ for
some $t\in\reals\setminus\{0\}$, there is a non-zero element $\xi$ of the
domain of $H_{\omega_0+t\delta_{n_0}}$ (which coincides with the domain of $
H_{\omega_0}$) such that
\begin{equation}
\label{eq:xi-eigenvalue}
  H_{\omega_0+t\delta_{n_0}}\xi=r_{k_0}(\omega_0)\xi\,.
\end{equation}
From (\ref{eq:main-difference-cases}), it is clear that both $\xi$ and
$\psi$ satisfy the difference equation
\begin{equation*}
  (\tau u)(n)=r_{k_0}(\omega_0)u(n)
\end{equation*}
for all $n$ such that $n_1+1<n<n_2-1$ and $n\ne n_0$. So, by
(\ref{eq:green}), $W_n(\xi,\psi)$ is constant for all $n$ such that
$n_0\le n<n_2$. Now, when $n_2$ is finite, both $\xi$ and
$\psi$ satisfy the difference equation (see
(\ref{eq:main-difference-cases}))
\begin{equation*}
  a(n-1)u(n-1)+ \omega(n)u(n)=r_{k_0}(\omega_0)u(n)\,,\quad
  \text{for } n=n_2-1\,.
\end{equation*}
This implies that $W_{n_2-2}((\xi,\psi))=0$, so the constant
$W_n(\xi,\psi)$, for all $n$ such that $n_0\le n<n_2-1$, is in fact
zero. If $n_2$ is infinite, then, from what was said in
Section~\ref{sec:preliminaries} (see
Remark~\ref{rem:self-adjoint-wronskian}) one concludes that
$W_n(\xi,\psi)=0$ for all $n\ge n_0$. Therefore, in both cases, $n_2$
finite or infinite, there exists $c\in\complex$ such that
$\xi(n)=c\psi(n)$ for all $n$ such that $n_0\le n<n_2-1$. This implies
that $\xi$ cannot satisfy (\ref{eq:xi-eigenvalue}) for $t\ne 0$ when
$\psi$ is an eigenvector with $\psi(n_0)\ne 0$. If $\psi(n_0)=0$, then
one may repeat the reasoning above for $n_0+1$, since, in this case,
it follows from (\ref{eq:main-difference-cases}) that $\psi(n_0+1)\ne
0$. Thus we assert that either
\begin{equation}
\label{eq:at-least-one-one}
\mu_{n_1+1}^{\omega_0+t\delta_{n_0}}(\{r_{k_0}(\omega_0)\})=0\,,\qquad
 \forall t\in\reals\setminus\{0\}\,,
\end{equation}
or
\begin{equation}
\label{eq:at-least-one-two}
  \mu_{n_1+1}^{\omega_0+s\delta_{n_0+1}}(\{r_{k_0}(\omega_0)\})=0\,,
\qquad
  \forall s\in\reals\setminus\{0\}\,,
\end{equation}
for any $\omega_0\in\mathcal{Q}^{r_{k_0}}$.  Let $\mathcal{Q}_1$ be
the set of $\omega\in\mathcal{Q}^{r_{k_0}}$ such that
(\ref{eq:at-least-one-one}) holds, and
$\mathcal{Q}_2=\mathcal{Q}^{r_{k_0}}\setminus\mathcal{Q}_1$. Thus we
have the partition
$\mathcal{Q}^{r_{k_0}}=\mathcal{Q}_1\cup\mathcal{Q}_2$.  Notice that,
if $\psi(n_0)= 0$, then $\psi$ is an eigenvector of
$H_{\omega_0+t\delta_{n_0}}$ for all $t\in\reals$. Thus, for any
$\omega_0\in\mathcal{Q}_2$,
\begin{equation}
  \label{eq:mu-positive}
  \mu_{n_1+1}^{\omega_0+t\delta_{n_0}}(\{r_{k_0}(\omega_0)\})>0\qquad
  \forall t\in\reals\,.
\end{equation}
Let us denote by $\chi_{\mathcal{A}}$ the characteristic function of
$\mathcal{A}$, that is,
\begin{equation}
  \label{eq:def-characteristic}
  \chi_{\mathcal{A}}(\omega)=
  \begin{cases}
    1 & \text{ if }\omega\in\mathcal{A} \\
    0 & \text{ if }\omega\not\in \mathcal{A}\,.
  \end{cases}
\end{equation}
Since $\mu_{n_1+1}^{\omega}(\{r\})$ is a measurable function of
$\omega\in\Omega$ for any fixed $r\in\reals$ (see
\cite[Sec.\,5.3]{MR1102675}), we know that
$\mu_{n_1+1}^{\omega+t\delta_{n_0}+s\delta_{n_0+1}}(\{r\})$ is a
measurable function of $(t,s)\in \reals^2$ (see
\cite[Thm.\,7.5]{MR0924157}) for any fixed
$\omega\in\Omega$. Therefore, using (\ref{eq:property-of-r}), one
establishes that
\begin{equation*}
  \chi_{\mathcal{Q}^{r_{k_0}}}^{-1}(\{1\})=
\{(t,s)\in\mathbb{R}^2:\mu_{n_1+1}^{\omega+t\delta_{n_0}+s\delta_{n_0+1}}
(\{r_{k_0}(\omega)\})>0\}
\end{equation*}
is measurable. Hence
\begin{equation*}
  (t,s)\to\chi_{\mathcal{Q}^{r_{k_0}}}(\omega+t\delta_{n_0}+s\delta_{n_0+1})
\end{equation*}
is a measurable function for any fixed
$\omega\in\Omega$. Thus, by Fubini
\begin{equation*}
\begin{split}
  &\int_{\reals^2}
\chi_{\mathcal{Q}^{r_{k_0}}}(\omega+t\delta_{n_0}+s\delta_{n_0+1})
d(p_{n_o}\times p_{n_0+1})(t,s)\\
&=\int_{\reals}\left[\int_{\reals}
\chi_{\mathcal{Q}^{r_{k_0}}}(\omega+t\delta_{n_0}+s\delta_{n_0+1})
dp_{n_o}(t)\right] dp_{n_0+1}(s)\,.
\end{split}
\end{equation*}
The following equality holds
\begin{equation}
  \label{eq:int-equals-chi}
  \int_{\reals}
\chi_{\mathcal{Q}^{r_{k_0}}}(\omega+t\delta_{n_0}+s\delta_{n_0+1})
dp_{n_o}(t)=\chi_{\mathcal{Q}_2}(\omega+s\delta_{n_0+1})\,.
\end{equation}
When $\omega+s\delta_{n_0+1}\in\mathcal{Q}^{r_{k_0}}$,
(\ref{eq:int-equals-chi}) is verified using
(\ref{eq:at-least-one-one}), (\ref{eq:mu-positive}),
$p_{n_0}(\reals)=1$ and the continuity of $p_{n_0}$.
If $\omega+s\delta_{n_0+1}\not\in\mathcal{Q}^{r_{k_0}}$, then either
$\omega+t\delta_{n_0}+s\delta_{n_0+1}\not\in\mathcal{Q}^{r_{k_0}}$ for
every $t\in\reals$ and (\ref{eq:int-equals-chi}) follows, or there
exists $t_0\in\reals$ such that
$\omega+t_0\delta_{n_0}+s\delta_{n_0+1}\in\mathcal{Q}^{r_{k_0}}$. If
$\omega+t_0\delta_{n_0}+s\delta_{n_0+1}\in\mathcal{Q}_1$,
(\ref{eq:int-equals-chi}) follows from (\ref{eq:at-least-one-one}) and
continuity of $p_{n_0}$. The case
$\omega+t_0\delta_{n_0}+s\delta_{n_0+1}\in\mathcal{Q}_2$ is not
possible since (\ref{eq:mu-positive}) would imply
$\omega+s\delta_{n_0+1}\in\mathcal{Q}^{r_{k_0}}$.

Notice that $\mathcal{Q}_2$ does not need to be measurable and
nevertheless the equality (\ref{eq:int-equals-chi}) shows that
$\chi_{\mathcal{Q}_2}(\omega+s\delta_{n_0+1})$ is a measurable
function of $s$.
Hence
\begin{equation*}
  \int_{\reals^2}
\chi_{\mathcal{Q}^{r_{k_0}}}(\omega+t\delta_{n_0}+s\delta_{n_0+1})
d(p_{n_o}\times p_{n_0+1})(t,s)=
\int_{\reals}
\chi_{\mathcal{Q}_2}(\omega+s\delta_{n_0+1})
dp_{n_0+1}(s)=0
\end{equation*}
since the support of $\chi_{\mathcal{Q}_2}(\omega+s\delta_{n_0+1})$ is
only one point as a consequence of (\ref{eq:at-least-one-two}).
So we arrive at the conclusion that, for any fixed
$\omega\in\Omega$,
\begin{equation*}
  \mu_{n_1+1}^{\omega+t\delta_{n_0}+s\delta_{n_0+1}}
(\{r_{k_0}(\omega)\})=0
\end{equation*}
for $p_{n_o}\times p_{n_0+1}$-a. e. $(t,s)$. Note that, since
\begin{equation*}
  \sum_k\mu_{n_1+1}^{\omega}(\{r_{k}(\omega)\})
\ge\mu_{n_1+1}^{\omega}(\cup_kr_{k}(\omega))\,,
\end{equation*}
we actually have that
\begin{equation}
  \label{eq:mu-zero-ts}
  \mu_{n_1+1}^{\omega+t\delta_{n_0}+s\delta_{n_0+1}}
(\cup_kr_{k}(\omega))=0
\end{equation}
for any fixed $\omega\in\Omega$, for $p_{n_o}\times p_{n_0+1}$-a. e. $(t,s)$.

Now,
let $Q:=\{\omega\in\Omega:\mu_{n_1+1}^{\omega}
(\cup_kr_{k}(\omega))>0\}$ and assume that it is measurable. Then
\begin{equation*}
  \begin{split}
    &\mathbb{P}(Q)=\int_\Omega\chi_Q(\omega)d \mathbb{P}(\omega)\\
    &=\int_{\reals^{I\setminus\{n_0,n_0+1\}}}
    \left[\int_{\reals^2}
\chi_Q(\widetilde{\omega}+t\delta_{n_0}+s\delta_{n_0+1})d(p_{n_0}\times
      p_{n_0+1})(t,s)\right]\mathop{\times}_{n\in
      I\setminus\{n_0,n_0+1\}}
dp_n(\widetilde{\omega})\,,
  \end{split}
\end{equation*}
where $\omega=\widetilde{\omega}+t\delta_{n_0}+s\delta_{n_0+1}$ and we
have used Fubini's theorem. From (\ref{eq:mu-zero-ts})
and the definition of $Q$, we have
\begin{equation*}
  \chi_Q(\tilde{\omega}+t\delta_{n_0}+s\delta_{n_0+1})=0
\end{equation*}
for $p_{n_0}\times p_{n_0+1}$ a. e. $(t,s)$. Therefore $\mathbb{P}(Q)=0$.

Thus we have proven (\ref{eq:mu-is-zero}) with $\phi=\delta_{n_1+1}$. To
prove it for an arbitrary $\phi\in l^2(I,\complex)$ observe that
$\mu_\phi^\omega\prec\mu_{n_1+1}^\omega$ \cite[Sec.\,70
Thm.\,1]{MR1255973}.

B) The numbers $n_1,n_2$ are infinite.

It follows from \cite[Sec.\,72]{MR1255973} and
(\ref{eq:absolute-continuity}) (cf. \cite[Eq.\,2.141]{MR1711536}) that
$r$ is an eigenvalue of $H_\omega$ if and only if
$(\mu_{m}^\omega+\mu_{m+1}^\omega)(\{r\})>0$ for any
fixed $m\in\integers$. Thus, one can repeat the proof for A) with
$\mu_{m}^\omega+\mu_{m+1}^\omega$ instead of
$\mu_{n_1+1}^\omega$. Hence one proves that either
\begin{equation*}
  (\mu_{m}^\omega+\mu_{m+1}^\omega)(\cup_kr_k(\omega))=0
\end{equation*}
for $\mathbb{P}$ a. e. $\omega$, or the set of $\omega$ where the
equality above holds is not measurable. The proof is then completed by
recalling that, for all $\phi\in l^2(\integers,\complex)$,
$\mu_\phi^\omega\prec\mu_{m}^\omega+\mu_{m+1}^\omega$
(this follows as in the first part of the proof of \cite[Sec.\,70
Thm.\,1]{MR1255973} using \cite[Sec.\,72]{MR1255973}).
\end{proof}
\begin{theorem}
  \label{thm:measurable-function}
  Let $\{r_k\}_k$ be a finite or infinite sequence of measurable
  functions ($r_k:\Omega\to\reals$). The function $h:\Omega\to\reals$
  given by
  \begin{equation*}
    h(\omega):=\mu_\phi^\omega(\cup_kr_k(\omega))
  \end{equation*}
is measurable.
\end{theorem}
\begin{proof}
  Consider a simple function
  $s(\omega)=\sum_{j=1}^N\alpha_j\chi_{A_j}(\omega)$, where
  $\chi_{A_j}(\omega)$ is the characteristic function of $A_j$ (see
  (\ref{eq:def-characteristic})). Note that $A_j=s^{-1}(\{\alpha_j\})$
  and the sets $\{A_j\}_{j=1}^N$ form a partition of $\Omega$.

  Let $V\subset\reals$ be an open set. The set
  \begin{equation*}
    A:=\{\omega\in\Omega:\inner{\phi}{E_{H_\omega}(\{s(\omega)\})\phi}\in
    V\}
  \end{equation*}
is measurable. Indeed,
\begin{equation*}
  A=\cup_{j=1}^N\left[A_j\cap
\{\omega\in\Omega:\inner{\phi}{E_{H_\omega}(\{\alpha_j\})\phi}
\in V\}\right]
\end{equation*}
and each
$\{\omega\in\Omega:\inner{\phi}{E_{H_\omega}(\{\alpha_j\})\phi} \in
V\}$ is measurable (cf. the commentary after
\cite[Prop.\,V.3.1]{MR1102675}). Thus, the function
$\mu_\phi^\omega(s(\omega))$ is measurable. We approximate the
measurable function $r_1(\omega)$ by simple functions to obtain the
assertion of the theorem for $r_1(\omega)$.

Now, suppose that
 \begin{equation*}
    h_m(\omega):=\mu_\phi^\omega(\cup_{k=1}^mr_k(\omega))
  \end{equation*}
is a measurable function. Clearly,
\begin{equation*}
  h_{m+1}(\omega)=
  \begin{cases}
     h_m(\omega) & r_{m+1}(\omega)\in \cup_{k=1}^mr_k(\omega)\\
     h_m(\omega)+\mu_\phi^\omega(r_{m+1}(\omega)) & \text{otherwise}.
  \end{cases}
\end{equation*}
So from the
measurability of $h_m(\omega)$ and $\mu_\phi^\omega(r_{m+1}(\omega))$,
the measurability of $h_{m+1}(\omega)$ follows. By induction we prove
the assertion of the theorem for any finite sequence of measurable
functions $\{r_k\}_k$. The case of an infinite sequence is proven by
taking a pointwise limit w.r.t. $\omega\in\Omega$ of $h_m(\omega)$
when $m$ tends to $\infty$.
\end{proof}
Let $\sigma_p(H_\omega)$ denote the set of eigenvalues of the operator
$H_\omega$.
\begin{corollary}
  \label{cor:measurability-spectrum}
  If $H_\omega$ is measurable \cite[Def.\,V.3.1]{MR1102675}, then
  $h(\omega):=\mu_\phi^\omega(\sigma_p(H_\omega))$ is a measurable
  function.
\end{corollary}
\begin{proof}
  Since the operator $H_\omega$ is measurable, we can apply a result
  of \cite{MR1690090} and give a measurable enumeration of the points
  in $\sigma_p(H_\omega)$. Then the assertion follows from
  Theorem~\ref{thm:measurable-function}.
\end{proof}
\section{Applications to spectral theory}
We begin this section by stating an elementary result.
\begin{lemma}
  \label{lem:equivalence}
Let $\mu$ be a measure on $X$ and let
\begin{equation*}
  \gamma(\Delta):=\int_\Delta f(\lambda)d\mu(\lambda)\,,
\end{equation*}
where $f$ is a non-negative measurable function. Then
\begin{equation*}
  \gamma\sim\mu\quad\iff\quad\mu(\{\lambda\in X:f(\lambda)=0\})=0\,.
\end{equation*}
\end{lemma}
\begin{proof}
  ($\Leftarrow$) $\gamma$ is absolutely continuous w.r.t $\mu$ by
  definition. Now, assume $\gamma(\Delta)=0$, then $f(\lambda)=0$ for
  $\mu$-a. e. $\lambda$ on $\Delta$ and
  \begin{equation*}
    \mu(\Delta)=\mu(\Delta\setminus\{\lambda\in X:f(\lambda)=0\}) +
    \mu(\{\lambda\in X:f(\lambda)=0\})=0\,.
  \end{equation*}
  ($\Rightarrow$) If $\mu(\{\lambda\in X:f(\lambda)=0\})>0$, then
  $\gamma(\{\lambda\in X:f(\lambda)=0\})=0$, so the measures are not
  equivalent.
 \end{proof}
\begin{theorem}
  \label{thm:equivalence-j-l}
  Assume that at least one of the numbers $n_1,n_2$ is finite and that
  $I$ contains at least three integers. Fix any $n,m\in I$. It turns
  out that, for $\mathbb{P}$-a. e. $\omega$,
  \begin{equation*}
    \mu_{n}^{\omega}\sim\mu_{m}^{\omega}\,.
  \end{equation*}
\end{theorem}
\begin{proof}
  Let $n_1>-\infty$.  Under this assumption we
  proceed stepwise. Firstly, we show that
  $\mu_{n}^{\omega}\sim\mu_{n_1+1}^{\omega}$ for $n_1<
  n<n_2-1$. Secondly, it is proven that
  $\mu_{n_2-2}^{\omega}\sim\mu_{n_2-1}^{\omega}$ when $n_2$ is finite.

  In view of the first equation in (\ref{eq:mu-n-mu-semi}),
  $\mu_{n}^{\omega}\sim\mu_{n_1+1}^{\omega}$ if and only if (see Lemma
  \ref{lem:equivalence})
  \begin{equation*}
    \mu_{n_1+1}^{\omega}(\{\lambda:s_{n_1}(\lambda,n)=0\})=0\,,\qquad
  \end{equation*}
  for $\mathbb{P}$-a. e. $\omega$. Due to the initial conditions
  (\ref{eq:initial-c}) and (\ref{eq:initial-s}), it is straightforward
  to verify that the polynomial $s_{n_1+1}(\lambda,n)$ is completely
  determined by the sequences $\{a(k)\}_{k=n_1+1}^{n-1}$ and
  $\{\omega(k)\}_{k=n_1+1}^{n-1}$. Now, the finite sequence
  $\{\lambda_k(\omega)\}_k$ of zeros of $s_{n_1+1}(\lambda,n)$
  satisfies the conditions imposed on the sequence
  $\{r_k(\omega)\}_k$ in the statement of
  Theorem~\ref{thm:spectral-measure-sequence} when $n_0\ge n$. By
  applying Theorem~\ref{thm:spectral-measure-sequence} and
  \ref{thm:measurable-function}, one completes the first step. Now,
  suppose that $n_2$ is finite, and use the second equation in
  (\ref{eq:mu-n-mu-semi}) to express $\mu_{n_2-2}^{\omega}$. The
  polynomial involved here, $c_{n_2}(\lambda,n_2-2)$, is completely
  determined by $a(n_2-2)$ and $\omega(n_2-1)$. The only root of this
  polynomial, satisfies the conditions imposed on the sequence
  $\{r_k(\omega)\}_k$ in
  Theorem~\ref{thm:spectral-measure-sequence} taking $n_0<n_2-2$.

  The statement of the theorem is completely proven after noticing
  that, when $n_1$ is not finite, one repeats the reasoning above,
  with $n_1$, $n_2$, $s_{n_1+1}(\lambda,n)$, $c_{n_2}(\lambda,n_2-2)$
  replaced by $n_2$, $n_1$, $c_{n_2}(\lambda,n)$,
  $s_{n_1+1}(\lambda,n_1+2)$, respectively.
\end{proof}
\begin{remark}
  Theorem \ref{thm:equivalence-j-l} is proven in \cite{MR1779620} for
  the case of absolutely continuous probability distributions in a
  more general setting. Our approach is different. In particular we do
  not need Poltoratskii's theorem used in \cite{MR1779620}.
\end{remark}
\begin{remark}
  \label{rem:examples}
  One may construct self-adjoint Jacobi operators for which
  $\mu_{n_1+1}^{\omega}\not\sim\mu_{n_1+2n}^{\omega}$ for all
  $n\in\nats$ and fixed $\omega$. Indeed, as mentioned in
  \cite[Example\,1]{MR1971777} for $n_1$ finite and $n_2$ infinite,
  there are self-adjoint Jacobi matrices such that
  $\mu_{n_1+1}^{\omega}(\{0\})\ne 0$ and $s_{n_1+1}(0,n_1+2n)=0$. On
  the other hand, there exist Jacobi operators for which
  $\mu_{n}^{\omega}\sim\mu_{m}^{\omega}$  when $n$ and
  $m$ are sufficiently big. This is the case of the self-adjoint
  Jacobi operator studied in \cite{naboko} (see the proof of Corollary
  5.2 in \cite{naboko}).
\end{remark}
We now turn to the case, when neither of the numbers $n_1,n_2$ is
finite. Observe that by inserting (\ref{eq:absolute-continuity}) into
(\ref{eq:mu-n-mu-matrix}) one has
\begin{equation}
  \label{eq:mu-n-g-m-tr}
  \mu_{n}^{\omega}(\Delta)=\int_\Delta g_{(m,n)}(\lambda)
d(\mu_{m}^{\omega}+\mu_{m+1}^{\omega})(\lambda)\,,
\end{equation}
where
\begin{equation}
  \label{eq:def-g}
 g_{(m,n)}(\lambda):= \inner{\boldsymbol{R}_m(\lambda)
\begin{pmatrix}
c_{m+1}(\lambda,n)\\
s_{m+1}(\lambda,n)
\end{pmatrix}
}{
\begin{pmatrix}
c_{m+1}(\lambda,n)\\
s_{m+1}(\lambda,n)
\end{pmatrix}
}_{\complex^2}
\end{equation}
\begin{theorem}
  Assume that neither of the numbers $n_1,n_2$ is
finite. Fix any
  $k,l,m,n\in\integers$. For $\mathbb{P}$-a. e. $\omega$,
  \begin{equation*}
   \mu_{k}^{\omega}+\mu_{l}^{\omega}\sim\mu_{m}^{\omega}+\mu_{n}^{\omega}\,.
  \end{equation*}
\end{theorem}
\begin{proof}
  It follows from (\ref{eq:mu-n-g-m-tr}) that
  \begin{equation}
    \label{eq:mu-n-mu-m-ac}
   (\mu_{m}^{\omega}+\mu_{n}^{\omega})(\Delta)=\int_\Delta
   \left(g_{(m,m)}(\lambda)+g_{(m,n)}(\lambda)\right)
d(\mu_{m}^{\omega}+\mu_{m+1}^{\omega})(\lambda)\,.
  \end{equation}
  Let us show that
  $\mu_{m}^{\omega}+\mu_{n}^{\omega}\sim\mu_{m}^{\omega}+\mu_{m+1}^{\omega}$
  for $\mathbb{P}$-a. e. $\omega$. Due to (\ref{eq:mu-n-mu-m-ac}) and
  Lemma~\ref{lem:equivalence}, this will be done if one proves that
\begin{equation*}
  (\mu_{m}^{\omega}+\mu_{m+1}^{\omega})
  (\mathcal{B})=0\quad\text{for}\ \mathbb{P}\text{-}a. e.\omega\,,
\end{equation*}
where $\mathcal{B}:=\{\lambda:g_{(m,m)}(\lambda)=g_{(m,n)}(\lambda)=0\}$.

Observing that $g_{(m,n)}(\lambda)=0$ implies
\begin{equation*}
  \boldsymbol{R}_m(\lambda)\begin{pmatrix}
c_{m+1}(\lambda,n)\\
s_{m+1}(\lambda,n)
\end{pmatrix}=0\,,
\end{equation*}
we obtain
\begin{equation}
  \label{eq:a-b}
  b_m(\lambda)c_{m+1}(\lambda,n)s_{m+1}(\lambda,n)=
-a_m(\lambda)c_{m+1}^2(\lambda,n)
\end{equation}
for any $n,m\in\integers$. On the other hand, (\ref{eq:def-g}) and
(\ref{eq:initial-c}), (\ref{eq:initial-s}) imply
$g_{(m,m)}(\lambda)=a_m(\lambda)$. From (\ref{eq:def-g}) and
(\ref{eq:a-b}), it follows that 
\begin{equation*}
  g_{(m,n)}(\lambda)=s_{m+1}^2(\lambda,n)
-a_m(\lambda)(s_{m+1}^2\left(\lambda,n)+c_{m+1}^2(\lambda,n)\right)\,.
\end{equation*}
 So, assuming that
$g_{(m,m)}(\lambda)=g_{(m,n)}(\lambda)=0$ one obtains 
$g_{(m,n)}(\lambda)=s_{m+1}^2(\lambda,n)$. This implies that the set
$\mathcal{B}$ is finite and its elements satisfy the conditions
imposed on the elements of the sequence $\{r_k(\omega)\}_k$ used in
Theorem \ref{thm:spectral-measure-sequence}. That theorem and Theorem
\ref{thm:measurable-function} yield that
$\mu_{m}^{\omega}+\mu_{n}^{\omega}\sim\mu_{m}^{\omega}+\mu_{m+1}^{\omega}$. Now,
the claim of the theorem follows from Remark~\ref{rem:equivalence}.
\end{proof}
\begin{remark}
  In the case of absolutely continuous distributions, it is proven in
  \cite{MR1779620} the stronger statement $\mu_m^\omega\sim\mu_n^\omega$ for
  $\mathbb{P}$-a. e. $\omega$ and any $m,n\in\integers$.
\end{remark}
\begin{theorem}
\label{thm:matrix-eigenvalues-submatrix}
   Consider an interval $\tilde{I}$ such that
  $\widetilde{I}\subset I\setminus\{m,m+1\}$, where $n_1+1\le m\le
  n_2-2$. Let $H_\omega$ be the operator defined in Section
  \ref{sec:preliminaries} in $l^2(I,\complex)$ and
  $\widetilde{H}_\omega$ the operator defined analogously in
  $l^2(\widetilde{I},\complex)$. Then,
  \begin{equation*}
   \mathbb{P}(\{\omega\in\Omega:\sigma_p(H_\omega)\cap
\sigma_p(\widetilde{H}_\omega)\ne\emptyset\})=0\,.
  \end{equation*}
\end{theorem}
\begin{proof}
 Observe that
$\sigma_p(\widetilde{H}_\omega)$ does not depend on
$\omega(m),\omega(m+1)$. Thus, it follows from
  Theorems~\ref{thm:spectral-measure-sequence},
  \ref{thm:measurable-function} and Corollary
  \ref{cor:measurability-spectrum} that
  \begin{equation}
    \label{eq:eigenvalues-disjoint}
    \mu_\phi^\omega(\sigma_p(\widetilde{H}_\omega))=0
  \end{equation}
for $\mathbb{P}$-a. e. $\omega$.

If $n_1$ (or $n_2$) is finite, take $\phi=\delta_{n_1+1}$
($\phi=\delta_{n_2-1}$), and, taking into account that
$\lambda\in\sigma_p(H_\omega)$ if and only if
$\mu_{n_1+1}^\omega(\{\lambda\})>0$, the theorem follows from
(\ref{eq:eigenvalues-disjoint}).

Now, assume that both $n_1$, $n_2$
are infinite and choose consecutively $\phi=\delta_0$ and
$\phi=\delta_1$. Then
\begin{equation*}
  (\mu_{0}^\omega+\mu_{1}^\omega)(\sigma_p(\widetilde{H}_\omega))=0
\end{equation*}
for $\mathbb{P}$-a. e. $\omega$. Since
\begin{equation*}
  \sigma_p(H_\omega)=\{\lambda\in\reals:
(\mu_{0}^\omega+\mu_{1}^\omega)(\{\lambda\})>0\}
\end{equation*}
(see \cite[Eq.\,2.141]{MR1711536}), the result
follows.
\end{proof}
\begin{corollary}
  Assume that at least one of the numbers $n_1,n_2$ is infinite. Then
  Theorem~\ref{thm:matrix-eigenvalues-submatrix} holds with any
  $\widetilde{I}\subsetneq I$.
\end{corollary}
\begin{proof}
  Assume for example that $n_2=+\infty$ and $n_1$ finite. Choose
  $\widetilde{I}=I\setminus\{n_1+1\}$. It is known that
  $\sigma_p(H_\omega)\cap\sigma_p(\widetilde{H}_\omega)=\emptyset$ for
  every $\omega$ \cite{MR0344937}. If we take any other
  $\widetilde{I}\subsetneq I$, Theorems
  \ref{thm:spectral-measure-sequence} and
  \ref{thm:measurable-function} can be applied. The other cases are
  handled analogously.
\end{proof}
\begin{remark}
  A more general situation could be considered along the same
  lines. Indeed, assume the same conditions as in Theorem
  \ref{thm:matrix-eigenvalues-submatrix} and let
 $r_k(\widetilde{H}_\omega)$ be a measurable real valued function of
  $\omega$ determined by $\widetilde{H}_\omega$. Then
  \begin{equation*}
   \mathbb{P}(\{\omega\in\Omega:\sigma_p(H_\omega)\cap
  \cup_kr_k(\widetilde{H}_\omega)\ne\emptyset\})=0\,.
  \end{equation*}
  For example each $r_k$ could be a matrix entry, a moment or any
  other quantity associated to $\widetilde{H}_\omega$.
\end{remark}

\begin{acknowledgments}
  We thank D. Damanik and F. Gesztesy for useful comments and
  pertinent hints to the literature. Stimulating discussions with
  J. Breuer, H. Krueger and H. Schulz-Baldes are also acknowledged.
\end{acknowledgments}

\end{document}